\documentclass[10pt,a4paper,onecolumn]{IEEEtran}
\usepackage{amsmath,amsfonts,amssymb,amsthm,graphicx,float,subfigure}
\def\EE{\mathbb E}
\def\Tr{\mathrm{Tr}}
\def\SNR{\mathrm{SNR}}
\newtheorem{stat}{Statement}[section]
\newtheorem{thm}[stat]{Theorem}
\newtheorem{lem}[stat]{Lemma}
\newtheorem{prop}[stat]{Proposition}

\usepackage{epstopdf}

\begin{document}

\title{Communication Tradeoffs in Wireless Networks} % not sure about this title!!!
\author{Serj Haddad, Olivier L\'ev\^eque\\Information Theory Laboratory\\ School of Computer and Communication Sciences\\ EPFL. 1015 Lausanne, Switzerland}
\maketitle
\thispagestyle{plain}
\pagestyle{plain}

\begin{abstract}
We characterize the maximum achievable broadcast rate in a wireless network under various fading assumptions. Our result exhibits a duality between the performance achieved in this context by collaborative beamforming strategies and the number of degrees of freedom available in the network.
\end{abstract}


%%%%%%%%%%%%%%%%%%%%%%%%
\section{Introduction}
There is a vast body of literature on the subject of multiple-unicast communications in ad hoc wireless networks. Because of the inherent broadcast nature of wireless signals, managing the {\em interference} between the multiple source-destination pairs is a key issue and has led to various interesting proposals \cite{GuptaUnicast,AyferUnicast, ADT11, NGS09,CJ08,RCY00,NGJV12,MOMK14}. In some of these works, it appeared that the model considered for the fading environment may substantially impact the performance of the proposed communication schemes (see \cite{FMM09}). In particular, the {\em channel diversity}, both spatial and temporal, turns out to be a key parameter for the analysis of the various schemes.

In the present paper, we address an a priori much easier scenario (previously considered in \cite{GastparBC}). Instead of every source node willing to communicate each to a different destination node, we consider the {\em broadcast scenario}, where each source node wishes to send some piece of information to all the other nodes in the network. This situation is to be encountered e.g. when control signals carrying channel state information should be broadcasted to the whole network.  In this context, the broadcast nature of the wireless medium can only help relaying communications, so that the situation seems simpler to handle, if not trivial. What we show in the following is that even in this simpler scenario, the optimal communication performance highly depends on the nature of the wireless medium. The conclusions we draw put again channel diversity to the forefront. But whereas diversity was beneficial for establishing multiple parallel communication channels in the multiple-unicast scenario, it turns out that in the present case, diversity is on the contrary detrimental to a proper broadcasting of information. A {\em duality} is further established between the number of {\em degrees of freedom} available for multi-party communications and the {\em beamforming gain} of broadcast transmissions, which allows for a better dissemination of information. At one end, in a rich scattering environment, degrees of freedom are prominent, while beamforming is practically infeasible. At the other end, degrees of freedom become a scarce resource, while high beamforming gains can be achieved via collaborative transmissions.

Our analysis relies on the simplistic line-of-sight fading model for signal attenuation over distance, where signal amplitude attenuation is inversely proportional to distance and phase shifts are also proportional to distance. Yet, this model, along with another parameter characterizing the {\em sparsity} of the network, allows to capture the different regimes mentioned above and to characterize the performance trade-offs. In addition, we would like to highlight here that despite the simplicity of the model, the mathematical analysis needed to establish the result on the maximum achievable broadcast rate in the network requires a precise and careful study of the spectral norm of unconventional random matrices, rarely studied in the mathematical literature.

%%%%%%%%%%%%%%%%%%%%%%%%
%%%%%%%%%%%%%%%%%%%%%%%%
%%%%%%%%%%%%%%%%%%%%%%%%
\section{Model} \label{sec:model}
There are $n$ nodes uniformly and independently distributed in a square of area $A=n^{\nu},~\nu > 0$. Every node wants to broadcast a different message to the whole network, and all nodes want to communicate at a common {\em per user} data rate $r_n$ bits/s/Hz. We denote by $R_n = n \, r_n$ the resulting {\em aggregate} data rate and will often refer to it simply as ``broadcast rate'' in the sequel. The broadcast capacity of the network, denoted as $C_n$, is defined as the maximum achievable aggregate data rate $R_n$. We assume that communication takes place over a flat fading channel with bandwidth $W$ and that the signal $Y_j[m]$ received by the $j$-th node at time $m$ is given by
$$
Y_j[m] = \sum_{k \in {\mathcal T}} h_{jk} \, X_k[m] + Z_j[m],
$$
where $\mathcal T$ is the set of transmitting nodes, $X_k [m]$ is the signal sent at time $m$ by node $k$ and $Z_j [m]$ is additive white circularly symmetric Gaussian noise (AWGN) of power spectral density $N_0/2$ Watts/Hz. We also assume a common average power budget per node of $P$ Watts, which implies that the signal $X_k$ sent by node $k$ is subject to an average power constraint $\EE(|X_k|^2) \le P$. In line-of-sight environment, the complex baseband-equivalent channel gain $h_{jk}$ between transmit node $k$ and receive node $j$ is given by
\begin{equation} \label{eq:model}
h_{jk} = \sqrt{G} \; \dfrac{\exp(2 \pi i r_{jk} / \lambda)}{r_{jk}},
\end{equation}
where $G$ is Friis' constant, $\lambda$ is the carrier wavelength, and $r_{jk}$ is the distance between node $k$ and node $j$. Let us finally define
$$
\SNR_s=\frac{GP}{N_0 W}n^{1-\nu},
$$
which is the SNR available for a communication between two nodes at distance $n^{\frac{\nu-1}{2}}$ in the network.

We focus in the following on the low SNR regime, by which we mean that $\SNR_s=n^{-\gamma}$ for some constant $\gamma>0$. This means that the power available at each node does not allow for a constant rate direct communication with a neighbor. This could be the case e.g., in a sensor network with low battery nodes, or in a sparse network (large $\nu$) with long distances between neighboring nodes.

In order to simplify notation, we choose new measurement units such that $\lambda=1$ and $G/(N_0 W)=1$ in these units. This allows us to write in particular that $\SNR_s=n^{1-\nu}P$.

%%%%%%%%%%%%%%%%%%%%%%%%
%%%%%%%%%%%%%%%%%%%%%%%%
%%%%%%%%%%%%%%%%%%%%%%%%
\section{Main result}
Before stating our main contribution, let us recall what is known for the multiple-unicast scenario \cite{JSAC}. In this case, the aggregated network throughput
scales as\footnote{up to logarithmic factors}
$$
T_n \sim \begin{cases} n \; \SNR_s & \text{if } (A/\lambda^2) \ge n^2\\ (\sqrt{A}/\lambda) \; \SNR_s & \text{if }n \le (A/\lambda^2) \le n^2\\ \sqrt{n} \; \SNR_s & \text{if } 1 \le  A/\lambda^2 \le n \end{cases}.
$$
Such an aggregate throughput is achieved by a hierarchical coooperative strategy involving network-wide distributed MIMO transmissions in the first two cases, while a simple multi-hopping strategy achieves the performance claimed in the third regime.

We therefore see that the wider the area is, the more degrees of freedom are available for communication in the network. The case where $A \sim n^2$ (corresponding to a sparse network of density $O(1/n)$) models the case where the phase shifts are large enough to ensure sufficient channel diversity and full degrees of freedom of MIMO transmissions. On the contrary, in the regime where $A \sim n$ (corresponding to a network of constant density), and even though this may seem surprising at first sight, phase shifts do not allow for efficient MIMO transmissions, so that multi-hopping becomes the best way to transfer information across the network.

A totally different scenario awaits us in the broadcast case. Our main result is the following: the aggregate broadcast rate scales as
$$
R_n \sim \begin{cases}  \min\{\SNR_s, 1\} & \text{if } (A/\lambda^2) \ge n^2\\ \min \left\{ \bigg( \dfrac{n}{\sqrt{A}/\lambda} \bigg) \,  \SNR_s, 1 \right\} & \text{if } 1 \le (A/\lambda^2) \le n^2 \end{cases}
$$
and is achieved by a simple broadcast transmission in the first case and by a multi-stage beamforming strategy in the second case. The performance is further capped at 1, which means that such beamforming gains can only be obtained at low SNR.

We see here for a sparse network of density $O(1/n)$  (regime where $A \sim n^2$), no particular beamforming gain can be obtained, while the beamforming gain increases as the network gets denser and denser. Let us mention here that the result where the network is of constant density ($A \sim n$) has been previously established in \cite{HLarxiv1}.

A final observation shows the duality of the two previous results: in the regime where $A/\lambda^2 \ge n$ (that is, for networks of constant density or sparser) and at low SNR, we have
$$
\frac{T_n}{\SNR_s} \; \frac{R_n}{\SNR_s} = n
$$
which captures the fact that high beamforming gains can only be obtained at the expense of a reduced number of degrees of freedom (or reciprocally).

%%%%%%%%%%%%%%%%%%%%%%%%%%
%%%%%%%%%%%%%%%%%%%%%%%%
%%%%%%%%%%%%%%%%%%%%%%%%
%%%%%%%%%%%%%%%%%%%%%%%%
\section{Broadcasting Strategies in Different Regimes} \label{sec:scheme}
First note that under the LOS model \eqref{eq:model} and the assumptions made in the Section \ref{sec:model}, a simple time division scheme achieves a broadcast (aggregate) rate $R_n$ of order $\min(\SNR_s,1)$. Indeed, a rate of order $1$ is obviously achieved at high SNR\footnote{We coarsely approximate $\log P$  by $1$ here!}. At low SNR (i.e.~when $\SNR_s \sim n^{-\gamma}$ for some $\gamma>0$), each node can spare power while the others are transmitting, so as to compensate for the path loss of order $1/n^{\nu}$ between the source node and other nodes located at distance at most $\sqrt{2n^{\nu}}$, leading to a broadcast rate of order $R_n \sim \log(1+ n P/n^{\nu}) \sim n^{1-\nu}P=\SNR_s$. 

In the following, we will see that, at low SNR, while the described simple TDMA based broadcast scheme is order-optimal for $A \ge n^2$, it is not optimal for sparse networks with area $A<n^2$ ($\nu<2$) (for simplicity, as stated in Section \ref{sec:model}, we take $\lambda=1$). On the other hand, the back-and-forth beamforming scheme, presented in \cite{HL15}, proves to be order-optimal for $A \le n^2$. 

As described in \cite{HL15}, the back-and-forth beamforming scheme involves source nodes taking turns to broadcast their messages. Each transmission is followed by a series of network-wide back-and-forth transmissions that reinforce the strength of the signal, so that at the end, every node is able to decode the message sent from the source. The reason why back-and-forth transmissions are useful for \textbf{small area networks}/\textbf{dense networks} is that in line-of-sight environment, nodes are able to (partly) align the transmitted signals so as to create a significant beamforming gain for each transmission (whereas this would not be the case in \textbf{high scattering environment}/\textbf{sparse networks} with i.i.d.~fading coefficients). In short, the back-and-forth beamforming scheme is split into two phases:\\

\textbf{Phase 1. Broadcast Transmission.} The source node broadcasts its message to the whole network. All the nodes receive a noisy version of the signal, which remains undecoded. This phase only requires one time slot.\\

\textbf{Phase 2. Back-and-Forth Beamforming with Time Division.} Upon receiving the signal from the broadcasting node, nodes start multiple back-and-forth beamforming transmissions between the two halves of the network to enhance the strength of the signal. Although this simple scheme probably achieves the optimal performance claimed in Theorem \ref{thm:BF_TDMA} below, we lack the analytical tools to prove it. For this reason, we propose a time-division strategy, where clusters of size $M=\frac{n^{\nu/4}}{2c_1}\times \frac{n^{\nu/2}}{4}$ and separated by horizontal distance $d=\frac{n^{\nu/2}}{4}$ pair up for the back-and-forth transmissions. During each transmission, there are $\Theta\left(n^{\nu/4-\epsilon}\right)$ cluster pairs operating in parallel, so $\Theta(n^{1-\epsilon})$ nodes are communicating in total. The number of rounds needed to serve all nodes must therefore be $\Theta(n^{\epsilon})$.

\begin{figure}[t]
\centering
\includegraphics[scale=0.5]{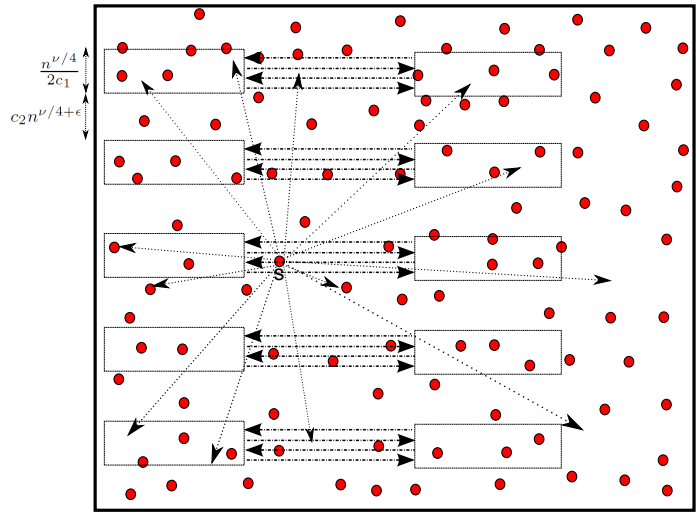}
\caption{$\sqrt{n}\times \sqrt{n}$ network divided into clusters of size $M=\frac{n^{\nu/4}}{2c_1}\times \frac{n^{\nu/2}}{4}$.
Two clusters of size $M$ placed on the same horizontal line and separated by distance $d=\frac{{n^{\nu/2}}}{4}$ pair up and start back-and-forth beamforming. The vertical separation between adjacent cluster pairs is $c_2n^{\nu/4+\epsilon}$.}
\label{fig:1}
\end{figure}  

After each transmission, the signal received by a node in a given cluster is the sum of the signals coming from the facing cluster, of those coming from other clusters, and of the noise. We assume a sufficiently large vertical distance $c_2 n ^{\nu/4+\epsilon}$ separating any two cluster pairs. We show below that the broadcast rate between the operating clusters is $\Theta(n^{2-\frac{3\nu}{2}}P)=\Theta(n^{1-\frac{\nu}{2}}\SNR_s)$. %????
Since we only need $\Theta(n^{\epsilon})$ number of rounds to serve all clusters, phase 2 requires $\Theta(n^{-2+\frac{3\nu}{2}+\epsilon}P^{-1})$ time slots. As such, back-and-forth beamforming achieves a broadcast rate of $\Theta(n^{2-\frac{3\nu}{2}-\epsilon}P)=\Theta(n^{1-\frac{\nu}{2}-\epsilon}\SNR_s)$ bits per time slot. In view of the described scheme, we are able to state the following result.

\begin{thm} \label{thm:BF_TDMA}
For any $\epsilon>0$, $0<\nu<2$, and $P=O(n^{-2+\frac{3\nu}{2}})$, the following broadcast rate
\begin{align*}
R_n=\Omega\left( n^{2-\frac{3\nu}{2}-\epsilon} P\right)=\Omega\left( n^{1-\frac{\nu}{2}-\epsilon} \SNR_s\right)
\end{align*}
is achievable with high probability\footnote{that is, with probability at least $1-O\left(\frac{1}{n^p}\right)$ as $n \to \infty$, where the exponent $p$ is as large as we want.} in the network. As a consequence, when $P = \Omega(n^{-2+\frac{3\nu}{2}})$, a broadcast rate $R_n = \Omega(n^{-\epsilon})$ is achievable with high probability.
\end{thm}

Before proceeding with the proof of the theorem, the following lemma provides an upper bound on the probability that the number of nodes inside each cluster deviates from its mean by a large factor. The proof is provided in the Appendix.

\begin{lem} \label{lem:cluster_size}
Let us consider a cluster of area $M$ with $M=n^\beta$ for some $\nu-1 < \beta <\nu$. The number of nodes inside each cluster is then between $((1-\delta)Mn^{1-\nu},\,(1+\delta)Mn^{1-\nu})$ with probability larger than $1-\frac{n^{\nu}}{M}\exp(-\Delta(\delta)Mn^{1-\nu})$ where $\Delta(\delta)$ is independent of $n$ and satisfies $\Delta(\delta)>0$ for $\delta>0$.
\end{lem}

As shown in Fig. \ref{fig:1}, two clusters of size $M=\frac{n^{\nu/4}}{2c_1}\times \frac{n^{\nu/2}}{4}$ placed on the same horizontal line and separated by distance $d=\frac{{n^{\nu/2}}}{4}$ form a cluster pair. During the back-and-forth beamforming phase, there are many cluster pairs operating simultaneously. Given that the cluster width is $\frac{n^{\nu/4}}{2 c_1}$ and the vertical separation between adjacent cluster pairs is $c_2n^{\nu/4+\epsilon}$, there are
$$
N_C=\frac{n^{\nu/2}}{\frac{n^{\nu/4}}{2 c_1}+c_2n^{\nu/4+\epsilon}} = \Theta \left( n^{\nu/4-\epsilon} \right)
$$
cluster pairs operating at the same time. Let $\mathcal{R}_i$ and $\mathcal{T}_i$ denote the receiving and the transmitting clusters of the $i$-th cluster pair, respectively.

Two key ingredients for analyzing the multi-stage back-and-forth beamforming scheme are given in Lemma \ref{lem:Full_Beamforming} and Lemma \ref{lem:Interference}. The proofs are presented in the Appendix. 

\begin{lem} \label{lem:Full_Beamforming}
The maximum beamforming gain between the two clusters of the $i$-th cluster pair can be achieved by using a compensation of the phase shifts at the transmit side which is proportional to the horizontal positions of the nodes. More precisely, there exist a constant $c_1>0$ (remember that $c_1$ is inversely proportional to the width of cluster $i$) and a constant $K_1>0$ such that the magnitude of the received signal at node $j \in \mathcal{R}_i$ is lower bounded with high probability by
$$
\left|\sum_{k\in \mathcal{T}_i} \frac{\exp(2\pi i (r_{jk}-x_k))}{r_{jk}}\right| \ge K_1 \frac{Mn^{1-\nu}}{d},
$$
where $x_k$ denotes the horizontal position of node $k$.
\end{lem} 

\begin{lem} \label{lem:Interference}
For every constant $K_2>0$, there exists a sufficiently large separating constant $c_2>0$ such that the magnitude of interfering signals from the simultaneously operating cluster pairs at node $j\in\mathcal{R}_i$ is upper bounded with high probability by
$$
\left|\sum_{\substack{l=1\\ l\neq i}}^{N_C}\sum_{k\in \mathcal{T}_l} \frac{\exp(2\pi i (r_{jk}-x_k))}{r_{jk}}\right| \le K_2\,\frac{Mn^{1-\nu}}{d \, n^{\epsilon}}\,\log n.
$$
\end{lem}

\begin{proof}[Proof of Theorem \ref{thm:BF_TDMA}]
The first phase of the scheme results in noisy observations of the message $X$ at all nodes, which are given by
\begin{align*}
Y_k^{(0)}=\sqrt{\SNR_k}\,X+Z_k^{(0)},
\end{align*}
where $\EE(|X|^2)=\EE(|Z_k^{(0)}|^2)=1$ and $\SNR_k$ is the signal-to-noise ratio of the signal $Y_k^{(0)}$ received at the $k$-th node. In what follows, we drop the index $k$ from $\SNR_k$ and only write $\SNR=\min_k\{\SNR_k\}$. Note that it does not make a difference at which side of the cluster pairs the back-and-forth beamforming starts or ends. Hence, assume the left-hand side clusters ignite the scheme by amplifying and forwarding the noisy observations of $X$ to the right-hand side clusters. The signal received at node $j\in\mathcal{R}_i$ is given by
\begin{equation}\label{Yj:1}
Y_j^{(1)} = \sum_{l=1}^{N_C}\sum_{k\in\mathcal{T}_l} \frac{\exp(2\pi i (r_{jk}-x_k))}{r_{jk}} A Y_k^{(0)} + Z_j^{(1)}
\end{equation}
where $A$ is the amplification factor (to be calculated later) and $Z_j^{(1)}$ is additive white Gaussian noise of variance $\Theta(1)$. We start by applying Lemma \ref{lem:Full_Beamforming} and Lemma \ref{lem:Interference} to lower bound
\begin{align*}
\left|\sum_{l=1}^{N_C}\sum_{k\in\mathcal{T}_l} \frac{\exp(2\pi i (r_{jk}-x_k))}{r_{jk}}\right|&\geq \left|\sum_{k\in\mathcal{T}_i} \frac{\exp(2\pi i (r_{jk}-x_k))}{r_{jk}}\right| - \left|\sum_{\substack{l=1\\ l\ne i}}^{N_C}\sum_{k\in\mathcal{T}_l} \frac{\exp(2\pi i (r_{jk}-x_k))}{r_{jk}}\right| \\
&\ge \left(K_1-K_2\frac{\log n}{n^{\epsilon}}\right) \frac{Mn^{1-\nu}}{d}=\Theta\left(\frac{Mn^{1-\nu}}{d}\right).
\end{align*}
For the sake of clarity, we can therefore approximate\footnote{We make this approximation to lighten the notation and make the exposition clear, but needless to say, the whole analysis goes through without the approximation; it just becomes barely readable.} the expression in \eqref{Yj:1} as follows
\begin{align*}
Y_j^{(1)} & = \sum_{l=1}^{N_C}\sum_{k\in\mathcal{T}_l} \frac{\exp(2\pi i (r_{jk}-x_k))}{r_{jk}} A \sqrt{\SNR_k}\,X + \sum_{l=1}^{N_C}\sum_{k\in\mathcal{T}_l} \frac{\exp(2\pi i (r_{jk}-x_k))}{r_{jk}} A Z_k^{(0)} + Z_j^{(1)}\\
& \simeq\frac{A Mn^{1-\nu}}{d} \sqrt{\SNR}\,X + \frac{A\sqrt{N_C Mn^{1-\nu}}}{d} Z^{(0)} + Z_j^{(1)}\\ &=\frac{A Mn^{1-\nu}}{d} \sqrt{\SNR}\,X + \frac{A Mn^{1-\nu}}{d}\sqrt{\frac{N_C}{Mn^{1-\nu}}} Z^{(0)} + Z_j^{(1)},
\end{align*}
where
$$
Z^{(0)}=\frac{d}{\sqrt{N_C  Mn^{1-\nu}}} \sum_{l=1}^{N_C} \sum_{k\in\mathcal{T}_l} \frac{\exp(2\pi i (r_{jk}-x_k))}{r_{jk}} Z_k^{(0)}.
$$
Note that $\EE(|Z^{(0)}|^2)=\Theta(1)$. Repeating the same process $t$ times in a back-and-forth manner results in a final signal at node $j\in\mathcal{R}_i$ in the left or the right cluster (depending on whether $t$ is odd or even) that is given by 
\begin{align*}
Y_j^{(k)} & = \left(\frac{A Mn^{1-\nu}}{d}\right)^t \sqrt{\SNR}\,X + \left(\frac{A Mn^{1-\nu}}{d}\right)^{t}\sqrt{\frac{N_C}{Mn^{1-\nu}}} \, Z^{(0)}\\
&  + \ldots + \left(\frac{A Mn^{1-\nu}}{d}\right)^{t-s}\sqrt{\frac{N_C}{Mn^{1-\nu}}} \, Z^{(s)}+\ldots+Z_j^{(t)},
\end{align*}
where
$$
Z^{(s)}= \frac{d}{\sqrt{N_C Mn^{1-\nu}}}\sum_{b=1}^{N_C}\sum_{k\in\mathcal{T}_b} \frac{\exp(2\pi i (r_{jk}-x_k))}{r_{jk}} Z_k^{(s)}.
$$
Note again that $\EE(|Z^{(s)}|^2)=\Theta(1)$, and $Z_j^{(t)}$ is additive white Gaussian noise of variance $\Theta(1)$. Finally, note that Lemma \ref{lem:Interference} ensures an upper bound on the beamforming gain of the noise signals, i.e.,  
$$
\left|\sum_{l=1}^{N_C}\sum_{k\in\mathcal{T}_l} \frac{\exp(2\pi i (r_{jk}-x_k))}{r_{jk}}\right|
\le \left|\sum_{k\in\mathcal{T}_i} \frac{\exp(2\pi i (r_{jk}-x_k))}{r_{jk}}\right| + \left|\sum_{\substack{l=1\\ l\ne i}}^{N_C}\sum_{k\in\mathcal{T}_l} \frac{\exp(2\pi i (r_{jk}-x_k))}{r_{jk}}\right| \le \left(1+K_2\frac{\log n}{n^{\epsilon}}\right)\frac{Mn^{1-\nu}}{d}.
$$
(notice indeed that the first term in the middle expression is trivially upper bounded by $Mn^{1-\nu}/d$, as it contains $M$ terms, all less than $1/d$). Now, we want the power of the signal to be of order 1, that is:
\begin{align}\label{eq:signal_power}
& \EE\left(\left(\left(\frac{A Mn^{1-\nu}}{d}\right)^t \sqrt{\SNR}\,X\right)^2\right)
= \left(\frac{A Mn^{1-\nu}}{d}\right)^{2t} \SNR = \Theta(1)\\\nonumber
& \Rightarrow A = \Theta\left(\frac{d}{Mn^{1-\nu}}\,\SNR^{-\frac{1}{2t}}\right).
\end{align}
Since at each round of TDMA cycle there are $\Theta\left(N_C Mn^{1-\nu}\right)=\Theta\left(n^{1-\epsilon}\right)$ nodes transmitting, then every node will be active $\Theta\left(\frac{N_C Mn^{1-\nu}}{n}\right)$ fraction of the time. As such, the amplification factor is given by
$$
A=\Theta\left(\sqrt{\frac{n^{\nu}}{N_C M}\tau P}\right),
$$
where $\tau$ is the number of time slots between two consecutive transmissions, i.e. every $\tau$ time slots we have one transmission. Therefore, we have
\begin{align*}
&A = \Theta\left(\frac{d}{Mn^{1-\nu}}\,\SNR^{-\frac{1}{2t}}\right)=\Theta\left(\sqrt{\frac{n^{\nu}}{N_C M}\tau P}\right)\\
& \Rightarrow \tau = O\left(\frac{1}{P} \left(\frac{d}{Mn^{1-\nu}}\right)^2n^{-\epsilon}\SNR^{-1/t}\right).
\end{align*}
We can pick the number of back-and-forth transmissions $t$ sufficiently large to ensure that $\SNR^{-\frac{1}{t}}=O(n^{\epsilon})$, which results in $$\tau=O\left(\frac{1}{P} \left(\frac{d}{Mn^{1-\nu}}\right)^2\right)=O\left(\frac{1}{n^{2-\frac{3\nu}{2}}P}\right).$$  
Moreover, the noise power is given by
\begin{align*}
\sum_{s=0}^{t-1} \EE\left(\left(\left(\frac{A Mn^{1-\nu}}{d}\right)^{t-s} \sqrt{\frac{N_C}{Mn^{1-\nu}}}Z^{(s)}\right)^2\right) + \EE\left(\left(Z_j^{(t)}\right)^2\right) &\le t \, \EE\left(\left(\left(\frac{A Mn^{1-\nu}}{d}\right)^t \sqrt{\frac{N_C}{Mn^{1-\nu}}}Z^{(0)}\right)^2\right) + 1\\
&\le t \, \left(\frac{A Mn^{1-\nu}}{d}\right)^{2t}{\frac{N_C}{Mn^{1-\nu}}}+1\\
&\overset{(a)}{\le}t+1=\Theta(1),
\end{align*}
where $(a)$ is true if and only if $\SNR=\Omega\left(\frac{N_C}{Mn^{1-\nu}}\right)=\Omega(n^{\nu/2-1-\epsilon})$ (check eq. \eqref{eq:signal_power}), which is true: Distance separating any two nodes in the network is as most $\sqrt{2n^{\nu}}$, which implies that the $\SNR$ of the received signal at all the nodes in the network is $\Omega(n\tau P/n^{\nu})=\Omega\left(n^{\nu/2-1}\right)$.

Given that the required $\tau=O\left(\frac{1}{n^{2-3\nu/2}P}\right)$, we can see that for $P=O(n^{3\nu/2-2})$ the broadcast rate between simultaneously operating clusters is $\Omega(n^{2-3\nu/2}P)$. Finally, applying TDMA of $\frac{n}{N_C Mn^{1-\nu}}=\Theta(n^{\epsilon})$ steps ensures that $X$ is successfully decoded at all nodes and the broadcast rate $R_n=\Omega\left(n^{2-3\nu/2-\epsilon}P\right)$. This completes the proof of the theorem. 
\end{proof}

%%%%%%%%%%%%%%%%%%%%%%%%
%%%%%%%%%%%%%%%%%%%%%%%%
%%%%%%%%%%%%%%%%%%%%%%%%
\section{Optimality of the Scheme} \label{sec:opt}
We start with the general upper bound already established in \cite{HL15} on the broadcast capacity of wireless networks at low SNR, which applies to a general fading matrix $H$.
\begin{thm} \label{thm:cap}
Let us consider a network of $n$ nodes and let $H$ be the $n \times n$ matrix with $h_{jj}=0$ on the diagonal and $h_{jk}=$ the fading coefficient between node $j$ and node $k$ in the network. The broadcast capacity of such a network with $n$ nodes is then upper bounded by
$$
C_n \le P \, \Vert H \Vert^2
$$
where $P$ is the power available per node and $\Vert H \Vert$ is the spectral norm (i.e.~the largest singular value) of $H$.
\end{thm}
 
We now aim to specialize Theorem \ref{thm:cap} to line-of-sight fading, where the matrix $H$ is given by
\begin{equation} \label{eq:fading_matrix}
h_{jk} = \begin{cases} 0 & \text{if }j=k \\ \dfrac{\exp(2 \pi i r_{jk})}{r_{jk}} & \text{if } j \ne k \end{cases}
\end{equation}
The rest of the section is devoted to proving the proposition below which, together with Theorem \ref{thm:cap}, shows the asymptotic optimality of the back-and-forth beamforming scheme for \textbf{small area networks}/\textbf{dense networks} ($0<\nu<2$) and the asymptotic optimality of the simple TDMA based broadcast scheme for \textbf{high scattering environment}/\textbf{sparse networks} ($\nu\geq 2$) at low SNR and under LOS fading.

\begin{figure}[t]
 \centering
 \includegraphics[scale=0.7]{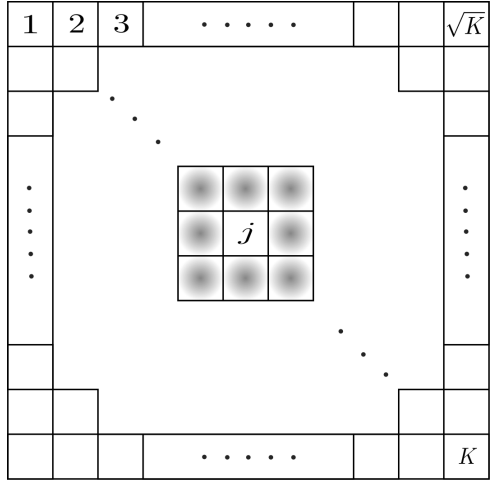}
 \caption{$\sqrt{n}\times\sqrt{n}$ network split into $K$ clusters and numbered in order. As such, $R_j=\{j-\sqrt{K}-1,j-\sqrt{K},j-\sqrt{K}+1,j-1,j,j+1, j+\sqrt{K}-1,j+\sqrt{K},j+\sqrt{K}+1\}$, which represents the center square containing the cluster $j$ and its $8$ neighbors (marked in shades).}
 \label{fig:2}
 \end{figure}
\begin{prop} \label{prop:ub_H}
Let $H$ be the $n \times n$ matrix given by \eqref{eq:fading_matrix}. For every $\varepsilon>0$, there exists a constant $c>0$ such that
$$
\Vert H \Vert^2 \le \begin{cases} 
n^{2-\frac{3\nu}{2}+\epsilon}P=n^{1-\frac{\nu}{2}+\epsilon}\text{ $\SNR_s$} & \text{if }0<\nu<2\\
n^{1-\nu+\epsilon}P=n^{\epsilon}\text{ $\SNR_s$} & \text{if }\nu\ge 2
\end{cases}
$$
with high probability as $n$ gets large.
\end{prop}

Analyzing directly the asymptotic behavior of $\Vert H \Vert$ reveals itself difficult. We therefore decompose our proof into simpler subproblems. The first building block of the proof is the following Lemma, which can be viewed as a generalization of the classical Ger\v{s}gorin discs' inequality.

\begin{lem} \label{lem:Gersgorin}
Let $B$ be an $n \times n$ matrix decomposed into blocks $B_{jk}$, $j,k=1,\ldots,K$, each of size $M \times M$, with $n=K M$. Then
$$
\Vert B \Vert \le \max \left\{ \max_{1 \le j \le K} \sum_{k=1}^K \Vert B_{jk} \Vert, \max_{1 \le j \le K} \sum_{k=1}^K \Vert B_{kj} \Vert \right\}
$$
\end{lem}

The proof of this Lemma is relegated to the Appendix. The second building block of this proof is the following lemma, the proof of which is also given in the Appendix.

\begin{lem}\label{lem:blockNorm}
Let $\widehat{H}$ be the $m \times m$ channel matrix between two square clusters of $m$ nodes distributed uniformly at random, each of area $A=m^{\nu}$, $\nu>0$, then
\begin{align*}
\Vert \widehat{H} \Vert^2 &\le \, \max \left\{\frac{m^{2+\epsilon}}{Ad}, \frac{m^{1+\epsilon}}{d^2} \right\}\\
&\le \begin{cases}\frac{m^{2+\epsilon}}{Ad} & \text{if } 0<\nu<2\\
\max\left\{\frac{m^{2+\epsilon}}{Ad},\frac{m^{1+\epsilon}}{d^2} \right\} & \text{if } \nu\ge 2
\end{cases}
\end{align*}
with high probability as $m$ gets large, where $2\sqrt{A} \le d \le A$ denotes the distance between the centers of the two clusters. 
\end{lem}
\begin{proof}[Proof of Proposition \ref{prop:ub_H}]
First we consider the case where $\nu\geq 2$. The strategy for the proof is now the following: in order to bound $\Vert H \Vert$, we divide the matrix into smaller blocks, apply Lemma \ref{lem:Gersgorin} and Lemma \ref{lem:blockNorm} in order to bound the off-diagonal terms $\Vert H_{jk} \Vert$. For the diagonal terms $\Vert H_{jj} \Vert$, we reapply Lemma \ref{lem:Gersgorin} and proceed in a recursive manner, until we reach small size blocks for which a loose estimate is sufficient to conclude.

Note that a network with area $A_0=n^{\nu}$ has a density of $n^{1-\nu}$. This means that a cluster of area $A_1=m_1n^{\nu-1}$ contains $m_1$ nodes with high probability. Let us therefore decompose the network into $K_1$ square clusters of area $m_1n^{\nu-1}$ with $m_1$ nodes each. Without loss of generality, we assume each cluster has exactly $m_1$ nodes and $K_1=n/m_1=A_0/A_1$. By Lemma \ref{lem:Gersgorin}, we obtain
\begin{equation} \label{eq:ub_H}
\Vert H \Vert \le \max \left\{ \max_{1 \le j \le K_1} \sum_{k=1}^{K_1} \Vert H_{jk} \Vert, \max_{1 \le j \le K_1} \sum_{k=1}^{K_1} \Vert H_{kj} \Vert \right\}
\end{equation}
where the $n \times n$ matrix $H$ is decomposed into blocks $H_{jk}$, $j,k=1,\ldots,K_1$, with $H_{jk}$ denoting the $m_1 \times m_1$ channel matrix between cluster number $j$ and cluster number $k$ in the network. Let us also denote by $d_{jk}$ the corresponding inter-cluster distance, measured from the centers of these clusters. Based on Lemma \ref{lem:blockNorm}, we obtain
$$
\Vert H_{jk} \Vert^2 \le \max \left\{\frac{m_1^{2+\epsilon}}{A_1d_{jk}}, \frac{m_1^{1+\epsilon}}{d_{jk}^2} \right\} \overset{(a)}{=} \frac{m_1^{1+\epsilon}}{d_{jk}^2}
$$
with high probability as $m_1 \to \infty$, where $(a)$ follows from the fact that $A_1/m_1=n^{\nu-1} \ge n^{\nu/2} \ge d_jk$, since $\nu\ge 2$ (equivalently, $\frac{m_1}{A_1}\le \frac{1}{d_{jk}}$).

Let us now fix $j \in \{1,\ldots,K_1\}$ and define $R_j = \{ 1 \le k \le K_1: d_{jk} < 2\sqrt{A_1} \}$ and $S_j = \{1 \le k \le K_1: d_{jk} \ge 2\sqrt{A_1} \}$ (see Fig. \ref{fig:2}). By the above inequality, we obtain
$$
\sum_{k=1}^{K_1} \Vert H_{jk} \Vert \le \sum_{k \in R_j} \Vert H_{jk} \Vert + \sqrt{n^{\epsilon}} \, \sum_{k \in S_j} \frac{\sqrt{m_1}}{d_{jk}}
$$
with high probability as $m_1$ gets large. Observe that as there are $8t$ clusters or less at distance $t\sqrt{A_1}$ from cluster $j$, so we obtain 
\begin{align*}
\sum_{k \in S_j} \frac{\sqrt{m_1}}{d_{jk}} & \le \sum_{t=2}^{\sqrt{K_1}} 8t \, \frac{\sqrt{m_1}}{t \sqrt{A_1}} = O  \left(\sqrt{\frac{K_1m_1}{A_1}} \right). 
\end{align*}
There remains to upper bound the sum over $R_j$. Observe that this sum contains at most 9 terms: namely the term $k=j$ and the 8 terms corresponding to the 8 neighboring clusters of cluster $j$.  It should then be observed that for each $k \in R_j$, $\Vert H_{jk} \Vert \le \Vert H(R_j) \Vert$, where $H(R_j)$ is the $9m_1 \times 9m_1$ matrix made of the $9 \times 9$ blocks $H_{j_1,j_2}$ such that $j_1,j_2 \in R_j$. Finally, this leads to
$$
\sum_{k=1}^{K_1} \Vert H_{jk} \Vert \le 9 \Vert H(R_j) \Vert + 8\sqrt{ n^{\epsilon}} \, \sqrt{\frac{K_1m_1}{A_1}}  
$$
Using the symmetry of this bound and \eqref{eq:ub_H}, we obtain
\begin{equation} \label{eq:recursion}
\Vert H \Vert \le 9 \, \max_{1 \le j \le K_1} \Vert H(R_j) \Vert + 8\sqrt{ n^{\epsilon}} \, \sqrt{\frac{K_1m_1}{A_1}}
\end{equation}
A key observation is now the following: For all $1\le j\le K_1$, the $9M \times 9M$ matrix $H(R_j)$ has exactly the same structure as the original matrix $H$. Therefore, without loss of generality, let us assume $\Vert H_1 \Vert=\max_{1 \le j \le K_1} \Vert H(R_j) \Vert=\Vert H(R_1) \Vert$. Finally, to bound $\Vert H_1 \Vert$, the same technique may be reused. This leads to the following recursive solution. 
\begin{align*}
\Vert H \Vert &= O\left( \Vert H_1 \Vert + \sqrt{ n^{\epsilon}} \, \sqrt{\frac{K_1m_1}{A_1}} \right)\\
&=O\left( \Vert H_2 \Vert + \sqrt{ n^{\epsilon}} \, \sqrt{\frac{K_2m_2}{A_2}}+\sqrt{ n^{\epsilon}}\sqrt{\frac{K_1m_1}{A_1}} \right)\\
&=O\left( \Vert H_l \Vert + \sqrt{ n^{\epsilon}} \sum_{t=1}^l \sqrt{\frac{K_t m_t}{A_t}} \right)\\
&=O\left(\Vert H_l \Vert + \sqrt{n^{\epsilon}}\sqrt{n^{1-\nu}} \sum_{t=1}^l \sqrt{K_t} \right),
\end{align*}
where $m_i$ denotes the number of nodes in a square cluster of area $A_i$. Moreover, $K_i=A_{i-1}/A_i=m_{i-1}/m_i$ denotes the number of square clusters of area $A_i$ and $m_i$ nodes in a square network of area $A_{i-1}$ containing $m_{i-1}$ nodes (note that $A_0=A=n^{\nu}$ and $m_0=n$). Finally, $\Vert H_i \Vert$ denotes the norm of the channel matrix of the network with square area $A_i$ and $m_i$ nodes. 

Note that we have a trivial bound on $\Vert H_i \Vert$. Apply for this the slightly modified version of the classical Ger\v{s}gorin inequality (which is nothing but the statement of Lemma \ref{lem:Gersgorin} applied to the case $M=1$):
$$
\Vert H_l \Vert \le \max\left\{ \max_{1 \le j \le m_l} \sum_{k=1}^{m_l} |(H_l)_{jk}|, \max_{1 \le j \le m_l} \sum_{k=1}^{m_l} |(H_l)_{jk}| \right\} = \max_{1 \le j \le m_l} \sum_{k=1 \atop k \ne j}^{m_l} \frac{1}{r_{jk}}
$$
For any $1 \le j \le m_l$, it holds with high probability that for $c$ large enough, 
$$
\sum_{k=1 \atop k \ne j}^{m_l} \frac{1}{r_{jk}} \le \sum_{t=1}^{\sqrt{m_l}} \frac{c\,t\,\log n}{t\,n^{\frac{\nu-1}{2}}} = O\left(\sqrt{m_l}\,n^{\frac{1-\nu}{2}} \log n\right)=O\left(n^{1-\nu}\sqrt{A_l}\log n\right),
$$
where the first inequality comes from the fact that at a distance $t\,n^{\frac{\nu-1}{2}}$ there are at most $c\,t$ clusters of area $n^{\nu-1}$ with at most $\log n$ nodes each. This implies that $\Vert H_l \Vert = O  \left(\sqrt{n^{\epsilon}}\,n^{1-\nu} \sqrt{A_l} \right)$ for any $\epsilon>0$. Therefore, we have
\begin{align*}
\Vert H \Vert &=O\left(\sqrt{n^{\epsilon}}\,n^{1-\nu} \sqrt{A_l} + \sqrt{n^{\epsilon}}\sqrt{n^{1-\nu}} \sum_{t=1}^l \sqrt{\frac{A_t}{A_{t-1}}} \right).
\end{align*}

Upon optimizing over the $A_i$'s, we get $A_i= n^{\nu-\frac{i}{l+1}}$. Note that $A_i$ is a decreasing function of $i$ and $A_0=n^{\nu}$. As such, for $\nu\ge 2$, we get the desired result
$$
\Vert H \Vert =O\left(n^{\frac{1-\nu}{2}}n^{\epsilon+\frac{1}{2\,(l+1)}}\right),
$$
where for any $\epsilon'>\epsilon$, we can pick $l$ large enough so as $\epsilon'<\epsilon+\frac{1}{2\,(l+1)}$ (notice that $\epsilon$ and $\epsilon'$ can be as small as we want). 

For $0<\nu<2$, we will take the following approach: We notice that a dense network can be seen as a superposition of sparse networks. In other words, we will look at a network with $n$ nodes uniformly and independently distributed over the area $n^{\nu}$, as the superposition of $n^{1-\nu/2}$ networks with $m=n^{\nu/2}$ nodes uniformly and independently distributed over the area $n^{\nu}=m^2$. Again, by Lemma \ref{lem:Gersgorin}, we obtain 
\begin{align*}
\Vert H \Vert &\le \max \left\{ \max_{1 \le j \le n^{1-\nu/2}} \sum_{k=1}^{n^{1-\nu/2}} \Vert H_{jk} \Vert, \max_{1 \le j \le n^{1-\nu/2}} \sum_{k=1}^{n^{1-\nu/2}} \Vert H_{kj} \Vert \right\}
\end{align*}
where the $n \times n$ matrix $H$ is decomposed into blocks $H_{jk}$, $j,k=1,\ldots,n^{1-\nu/2}$, with $H_{jk}$ denoting the $m \times m$ channel matrix between sparse network number $j$ and sparse network number $k$. Since each of these sparse networks has area $m^2$ with $m$ nodes, we can apply the upper bound we got for $\nu=2$, and $\forall \, j,k=1,\ldots,n^{1-\nu/2}$, obtain 
$$
\Vert H_{jk} \Vert = O\left(m^{-\frac{1}{2}+\epsilon}\right) = O\left(n^{-\frac{\nu}{4}+\frac{\epsilon}{2}}\right),
$$
which results in 
$$
\Vert H \Vert = O\left(n^{1-\frac{3\nu}{4}+\frac{\epsilon}{2}}\right).
$$
This finally proves Proposition \ref{prop:ub_H}.\end{proof}

%%%%%%%%%%%%%%%%%%%%%%%%
%%%%%%%%%%%%%%%%%%%%%%%%
%%%%%%%%%%%%%%%%%%%%%%%%
\section{Conclusion}
In this work, we characterize the broadcast capacity of a wireless network at low SNR in line-of-sight environment and under various assumptions regarding the network density. The result exhibits a dichotomy between sparse networks, where node collaboration can hardly help enhancing communication rates, and constant density networks, where significant gains can be obtained via collaborative beamforming.

\section{Acknowledgment}
S. Haddad's work is supported by Swiss NSF Grant Nr.~200020-156669.

%%%%%%%%%%%%%%%%%%%%%%%%
%%%%%%%%%%%%%%%%%%%%%%%%
%%%%%%%%%%%%%%%%%%%%%%%%
\appendix 

\begin{proof}[Proof of Lemma \ref{lem:cluster_size}]
The number of nodes in a given cluster is the sum of $n$ independently and identically distributed Bernoulli random variables $B_i$, with $\mathcal{P}(B_i=1)=M/n^{\nu}$. Hence
\begin{align*}
& \mathbb{P}\left(\sum_{i=1}^n B_i\geq (1+\delta)M n^{1-\nu}\right)\\
& =\mathbb{P}\left(\exp\left(s\sum_{i=1}^n B_i\right)\geq \exp(s(1+\delta)Mn^{1-\nu})\right)\\
& \leq \mathbb{E}^n(\exp(sB_1))\exp(-s(1+\delta)Mn^{1-\nu})\\
& = \left(\frac{M}{n^{\nu}}\exp(s)+1- \frac{M}{n^{\nu}}\right)^n\exp(-s(1+\delta)Mn^{1-\nu})\\
& \leq \exp(-Mn^{1-\nu}(s(1+\delta)-\exp(s)+1))=\exp(-Mn^{1-\nu}\Delta_+(\delta))
\end{align*}
where $\Delta_+(\delta)=(1+\delta)\log(1+\delta)-\delta$ by choosing $s=\log(1+\delta)$. The proof of the lower bound follows similarly by considering the random variables $-B_i$. The conclusion follows from the union bound.
\end{proof}
\begin{figure}
\centering
\includegraphics[scale=0.7]{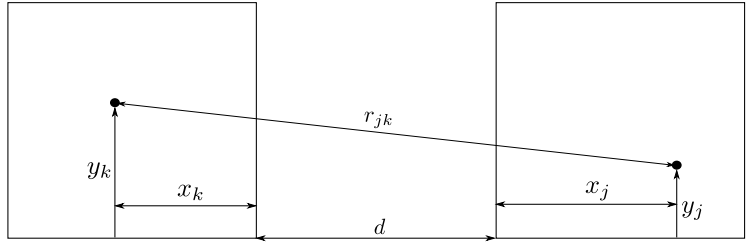}
\caption{Coordinate system.}
\label{fig:3}
\end{figure}  
\begin{proof}[Proof of Lemma \ref{lem:Full_Beamforming}]
We present lower and upper bounds on the distance $r_{jk}$ separating a receiving node $j\in \mathcal{R}_i$ and a transmitting node $k\in\mathcal{T}_i$. Denote by $x_j$, $x_k$, $y_j$, and $y_k$ the horizontal and the vertical positions of nodes $j$ and $k$, respectively (as shown in Fig. \ref{fig:3}). An easy lower bound on $r_{jk}$ is
$$
r_{jk}\geq x_k+x_j+d
$$
On the other hand, using the inequality $\sqrt{1+x} \le 1 + \frac{x}{2}$, we obtain
\begin{align*}
r_{jk}&=\sqrt{(x_k+x_j+d)^2+(y_j-y_k)^2}\\
&=\left(x_k+x_j+d\right)\sqrt{1+\frac{(y_j-y_k)^2}{(x_k+x_j+d)^2}}\\
&\leq x_k+x_j+d+\frac{(y_j-y_k)^2}{2d}\leq x_k+x_j+d+\frac{1}{2c_1^2}.
\end{align*}
Therefore, $$0\leq r_{jk}-x_k-x_j-d \leq \frac{1}{2c_1^2}.$$ After bounding $r_{jk}$, we can proceed to the proof of the lemma as follows:
\begin{align*}
\left|\sum_{k\in \mathcal{T}_i} \frac{\exp(2\pi i (r_{jk}-x_k))}{r_{jk}}\right|&=  \left|\sum_{k\in \mathcal{T}_i} \frac{\exp(2\pi i (r_{jk}-x_k-x_j-d))}{r_{jk}}\right|\\
&\geq \Re \left(\sum_{k\in \mathcal{T}_i} \frac{\exp(2\pi i (r_{jk}-x_k-x_j-d))}{r_{jk}}\right)\\
&\geq \sum_{k\in \mathcal{T}_i} \frac{\cos\left(\frac{\pi}{ c_1^2}\right)}{r_{jk}}\ge K_1\frac{Mn^{1-\nu}}{d},
\end{align*}
when the constant $c_1$ is chosen sufficiently large so that $\cos\left(\frac{\pi}{ c_1^2}\right)>0$.
\end{proof}

\begin{proof}[Proof of Lemma \ref{lem:Interference}]
There are $N_C$ clusters transmitting simultaneously. Except for the horizontally adjacent cluster of a given cluster pair ($i$-th cluster pair), all the rest of the transmitting clusters are considered as interfering clusters (there are $N_C-1$ of these). With high probability, each cluster contains $\Theta(Mn^{1-\nu})$ nodes. For the sake of clarity, we assume here that every cluster contains exactly $Mn^{1-\nu}$ nodes, but the argument holds in the general case. In this lemma, we upper bound the magnitude of interfering signals from the simultaneously interfering clusters at node $j\in\mathcal{R}_i$ as follows
\begin{align*}
\left|\sum_{\substack{l=1\\ l\neq i}}^{N_C}\sum_{k\in \mathcal{T}_l} \frac{\exp(2\pi i (r_{jk}-x_k))}{r_{jk}}\right| 
&\le \left|\sum_{\substack{l=1\\ l\neq i}}^{N_C}\sum_{k\in \mathcal{T}_l} \frac{\cos(2\pi (r_{jk}-x_k))}{r_{jk}}\right|+\left|\sum_{\substack{l=1\\ l\neq i}}^{N_C}\sum_{k\in \mathcal{T}_l} \frac{\sin(2\pi (r_{jk}-x_k))}{r_{jk}}\right|
\end{align*}
We only upper bound the first term (cosine terms) in the equation above as follows (we can upper bound the second term (sine terms) in exactly the same fashion):
\begin{align}\nonumber
\left|\sum_{\substack{l=1\\ l\neq i}}^{N_C}\sum_{k\in \mathcal{T}_l} \frac{\cos(2\pi (r_{jk}-x_k))}{r_{jk}}\right| &\leq \left|\sum_{\substack{l=1\\ l\neq i}}^{N_C}\sum_{k\in \mathcal{T}_l} \frac{\cos(2\pi (r_{jk}-x_k))}{r_{jk}}-\EE\left(\frac{\cos(2\pi (r_{jk}-x_k))}{r_{jk}}\right)\right|+ \left|\sum_{\substack{l=1\\ l\neq i}}^{N_C}\sum_{k\in \mathcal{T}_l}\EE\left(\frac{\cos(2\pi (r_{jk}-x_k))}{r_{jk}}\right)\right|\\\nonumber &\leq \left|\sum_{\substack{l=1\\ l\neq i}}^{N_C}\sum_{k\in \mathcal{T}_l} \frac{\cos(2\pi (r_{jk}-x_k))}{r_{jk}}-\EE\left(\frac{\cos(2\pi (r_{jk}-x_k))}{r_{jk}}\right)\right|+ \sum_{\substack{l=1\\ l\neq i}}^{N_C}\left|\sum_{k\in \mathcal{T}_l}\EE\left(\frac{\cos(2\pi (r_{jk}-x_k))}{r_{jk}}\right)\right|\\ &\leq \left|\sum_{\substack{l=1\\ l\neq i}}^{N_C}\sum_{k\in \mathcal{T}_l} \frac{\cos(2\pi (r_{jk}-x_k))}{r_{jk}}-\EE\left(\frac{\cos(2\pi (r_{jk}-x_k))}{r_{jk}}\right)\right|+ \sum_{\substack{l=1\\ l\neq i}}^{N_C}\left|\sum_{k\in \mathcal{T}'_l}\EE\left(\frac{\cos(2\pi (r_{jk}-x_k))}{r_{jk}}\right)\right|\label{eq:Expectation_X1}
\end{align}
where $\mathcal{T}'_l$ denotes the $l$-th interfering transmit cluster that is at a vertical distance of $l \left(\frac{n^{\nu/4}}{2c_1}+c_2 n^{\nu/4+\epsilon}\right)$ from the desired receiving cluster 
$\mathcal{R}_i$. Let us first bound the second term of \eqref{eq:Expectation_X1}. Denote by $X^{(l)}_k=(\cos(2\pi (r_{jk}-x_k)))/(r_{jk})\,\,\forall k\in\mathcal{T}'_l$. Note that $X^{(l)}_k$'s are independent and identically distributed. For any $k\in\mathcal{T}'_l$, we have
$$
|r_{jk}| = r_{jk} = \sqrt{(x_k+x_j+d)^2+(y_j-y_k)^2} \ge d = \frac{n^{\nu/2}}{4}$$
is a $C^2$ function and 
\begin{align*}
|r_{jk}'(y_k)|&=\left|\frac{\partial\,r_{jk} }{\partial{y_k}}\right| = \frac{|y_k-y_j|}{r_{jk}} \\
&\ge \frac{l \, c_2 \, n^{\nu/4+\epsilon}+(l-1) \, \frac{n^{\nu/4}}{2c_1}}{n^{\nu/2}}\\
&\ge {l \, c_2 \, n^{-\nu/4+\epsilon}}
\end{align*}
Moreover, $r_{jk}''$ changes sign at most twice. By the integration by parts formula, we obtain
\begin{align*}
\int_{{y_k}_0}^{{y_k}_1} \,dy_k \frac{\cos(2\pi  r_{jk})}{r_{jk}} & =
\int_{{y_k}_0}^{{y_k}_1} dy_k \, \frac{2 \pi  r'_{jk}}{2 \pi   r'_{jk}r_{jk}} \, \cos( 2\pi   r_{jk})\\
&=  \frac{ -\sin( 2 \pi  r_{jk})}{2 \pi   r'_{jk}r_{jk}} \bigg|_{{y_k}_0}^{{y_k}_1} + \frac{1}{2\pi }\int_{{y_k}_0}^{{y_k}_1} dy_k \, \frac{r_{jk}r''_{jk}+(r'_{jk})^2}{(r'_{jk}r_{jk})^2} \, \sin({ 2 \pi  r_{jk}})
\end{align*}
which in turn yields the upper bound
\begin{align*}
\left|\int_{{y_k}_0}^{{y_k}_1} dy_k \, \frac{\cos({2 \pi  r_{jk}})}{r_{jk}} \right|&\le \frac{1}{2 \pi} \, \Bigg( \frac{2}{\min_{y_k}\{|r'_{jk}||r_{jk}|\}} 
+ \int_{{y_k}_0}^{{y_k}_1} dy_k \, \frac{|r''_{jk}|}{(r'_{jk})^2|r_{jk}|} + \int_{{y_k}_0}^{{y_k}_1} dy_k \, \frac{1}{r_{jk}^2} \Bigg)\\
&\le \frac{1}{2 \pi} \, \Bigg( \frac{4}{l \, c_2 \, n^{\nu/4+\epsilon}} 
+ \frac{1}{\min_{y_k}\{|r_{jk}|\}}\int_{{y_k}_0}^{{y_k}_1} dy_k \, \frac{|r''_{jk}|}{(r'_{jk})^2} + \frac{|{{y_k}_1}-{{y_k}_0}|}{\min_{y_k}\{r_{jk}^2\}} \Bigg)\\
&\le \frac{1}{2 \pi} \, \Bigg( \frac{4}{l \, c_2 \, n^{\nu/4+\epsilon}} 
+ \frac{4}{l \, c_2 \, n^{\nu/4+\epsilon}} + \frac{2}{n^{3\nu/4}} \Bigg)
\le \frac{9/(2\pi)}{l \, c_2 \, n^{\nu/4+\epsilon}}.~~ \text{(take $\nu>2\epsilon$)}
\end{align*}
Therefore, for any $k\in\mathcal{T}'_l$,
\begin{align}\nonumber
\bigg|\mathbb{E}\left(X^{(l)}_k\right)\bigg|&= \left|\frac{4}{n^{\nu/2}}\int_0^{\frac{n^{\nu/2}}{4}}\,dx_k\frac{1}{|{y_k}_1-{y_k}_0|}\int_{{y_k}_0}^{{y_k}_1} dy_k \, \frac{\cos({2 \pi  r_{jk}})}{r_{jk}} \right|\\\nonumber
&\le \frac{4}{n^{\nu/2} \, |{y_k}_1-{y_k}_0|}\int_0^{\frac{n^{\nu/2}}{4}}\,dx_k \left|\int_{{y_k}_0}^{{y_k}_1} dy_k \, \frac{\cos({2 \pi  r_{jk}})}{r_{jk}} \right|\\
&\le \frac{9/(2\pi)}{|{y_k}_1-{y_k}_0| \, l \, c_2 \, n^{\nu/4+\epsilon}}
\le \frac{9c_1}{\pi c_2}\frac{1}{l \, n^{\nu/2+\epsilon}}=\frac{9c_1}{\pi c_2} \frac{1}{l \, d \, n^{\epsilon}}.\label{eq:Exp_1}
\end{align}
We further upper bound the first term in \eqref{eq:Expectation_X1} by using the Hoeffding's inequality \cite{Hoeffding}. Denote by $X_k=\frac{\cos({2 \pi  r_{jk}})}{r_{jk}}$, where $1\le k\le N_C Mn^{1-\nu}=\Theta\left(n^{1-\epsilon}\right)$. Note that $X_k$'s  are i.i.d. and integrable random variables that represent all nodes in all the interfering clusters. In other words, we have
$$
\left|\sum_{\substack{l=1\\ l\neq i}}^{N_C}\sum_{k\in \mathcal{T}_l} \frac{\cos(2\pi (r_{jk}-x_k))}{r_{jk}}-\EE\left(\frac{\cos(2\pi (r_{jk}-x_k))}{r_{jk}}\right)\right|=\left|\sum_{k=1}^{n^{1-\epsilon}} \left(X_k-\mathbb{E}(X_k)\right)\right|.
$$
We have $X_k \in [-1/d,1/d]$. As such, Hoeffding's inequality yields
\begin{align*}
\mathbb{P}\left(\frac{1}{n^{1-\epsilon}}\left|\sum_{k=1}^{n^{1-\epsilon}} \left(X_k-\mathbb{E} \left(X_k\right)\right)\right|>t\right)&\leq 2\, \exp\left(-\frac{n^{1-\epsilon} \, t^2}{2/d^2}\right)\\
&=2\,\exp\left(-\frac{1}{2}n^{1-\epsilon} \, d^2 \, t^2\right)\\
&\overset{(a)}{=}2\exp(-n^{\epsilon_1}),
\end{align*}  
where $(a)$ is true if $t=\frac{\sqrt{2n^{\epsilon+\epsilon_1-1}}}{d}$. Therefore, we have
\begin{align}\label{eq:Exp_2}
\left|\sum_{k=1}^{n^{1-\epsilon}} \left(X_k-\mathbb{E}(X_k)\right)\right| &\leq n^{1-\epsilon}t =\frac{\sqrt{2n^{1-\epsilon+\epsilon_1}}}{d}
\end{align}
with probability $\ge 1-2\exp(-n^{\epsilon_1})$. Combining \eqref{eq:Exp_1} and \eqref{eq:Exp_2}, we can upper bound \eqref{eq:Expectation_X1} as follows
\begin{align*}
\left|\sum_{\substack{l=1\\ l\neq i}}^{N_C}\sum_{k\in \mathcal{T}_l} \frac{\cos(2\pi (r_{jk}-x_k))}{r_{jk}}\right| &\le \frac{\sqrt{2n^{1-\epsilon+\epsilon_1}}}{d} + \sum_{l=1}^{N_C}\frac{9c_1}{\pi c_2} \frac{Mn^{1-\nu}}{l \, d \, n^{\epsilon}}\\
&\le \frac{\sqrt{2n^{1-\epsilon+\epsilon_1}}}{d} + \frac{9c_1}{\pi c_2} \frac{Mn^{1-\nu}}{d \, n^{\epsilon}}\log n.
\end{align*}
Note that for $M=\Theta\left(n^{3\nu/4}\right)$ and $\nu\leq 2-(\epsilon+\epsilon_1)$, we have
$$
\frac{\sqrt{2n^{1-\epsilon+\epsilon_1}}}{d} \le \frac{9c_1}{\pi c_2} \frac{Mn^{1-\nu}}{d \, n^{\epsilon}}\log n
$$
Finally, upper bounding the sine terms in the same fashion, we obtain
$$
\left|\sum_{\substack{l=1\\ l\neq i}}^{N_C}\sum_{k\in \mathcal{T}_l} \frac{\exp(2\pi i (r_{jk}-x_k))}{r_{jk}}\right|=O\left(\frac{Mn^{1-\nu}}{d \, n^{\epsilon}}\log n\right)
$$
with high probability (more precisely, with probability $\ge 1-4\, \exp(-n^{\epsilon_1})$), which concludes the proof.
\end{proof}

\begin{proof}[Proof of Lemma \ref{lem:Gersgorin}]
- Let us first consider the case where $B$ is a Hermitian and positive semi-definite matrix. Then $\Vert B \Vert=\lambda_{\max}(B)$, the largest eigenvalue of $B$. Let now $\lambda$ be an eigenvalue of $B$ and $u$ be its corresponding eigenvector, so that $\lambda u = Bu$. Using the block representation of the matrix $B$, we have
$$
\lambda \, u_j = \sum_{k=1}^K B_{jk} \, u_k, \quad \forall 1 \le j \le K
$$
where $u_j$ is the $j^{th}$ block of the vector $u$. Let now $j$ be such that $\Vert u_j \Vert = \max_{1 \le k \le K} \Vert u_k \Vert$. Taking norms and using the triangle inequality, we obtain
\begin{align*}
|\lambda| \, \Vert u_j \Vert & = \left\Vert \sum_{k=1}^K B_{jk} \, u_k \right\Vert \le \sum_{k=1}^K \Vert B_{jk} \, u_k \Vert\\
& \le \sum_{k=1}^K \Vert B_{jk} \Vert \, \Vert u_k \Vert \le \sum_{k=1}^K \Vert B_{jk} \Vert \, \Vert u_j \Vert
\end{align*}
by the assumption made above. As $u \not\equiv 0$, $\Vert u_j \Vert >0$, so we obtain
$$
|\lambda| \le \max_{1 \le j \le K} \sum_{k=1}^K \Vert B_{jk} \Vert
$$
As this inequality applies to any eigenvalue $\lambda$ of $B$ and $\Vert B \Vert=\lambda_{\max}(B)$, the claim is proved in this case.

- In the general case, observe first that $\Vert B \Vert^2=\lambda_{\max}(BB^{\dagger})$, where $BB^{\dagger}$ is Hermitian and positive semi-definite. So by what was just proved above,
$$
\Vert B \Vert^2 = \lambda_{\max}(BB^{\dagger}) \le \max_{1 \le j \le K} \sum_{k=1}^K \Vert (BB^{\dagger})_{jk} \Vert
$$
Now, $(BB^{\dagger})_{jk} = \sum_{l=1}^K B_{jl} B_{kl}^{\dagger}$ so
\begin{align*}
& \sum_{k=1}^K \Vert (BB^{\dagger})_{jk} \Vert = \sum_{k=1}^K \left\Vert \sum_{l=1}^K B_{jl} B_{kl}^{\dagger} \right\Vert\\
& \le \sum_{k=1}^K \sum_{l=1}^K \Vert B_{jl} \Vert \, \Vert B_{kl} \Vert \le \sum_{l=1}^K \Vert B_{jl} \Vert \, \max_{1 \le j \le K} \sum_{k=1}^K \Vert B_{kj} \Vert 
\end{align*}
and we finally obtain
$$
\Vert B \Vert^2 \le \left( \max_{1 \le j \le K} \sum_{l=1}^K \Vert B_{jl} \Vert  \right) \,  \left( \max_{1 \le j \le K} \sum_{k=1}^K \Vert B_{kj} \Vert \right)
$$
which implies the result, as $ab \le \max\{a,b\}^2$ for any two positive numbers $a,b$.
\end{proof}
\begin{figure}
\centering
\includegraphics[scale=1]{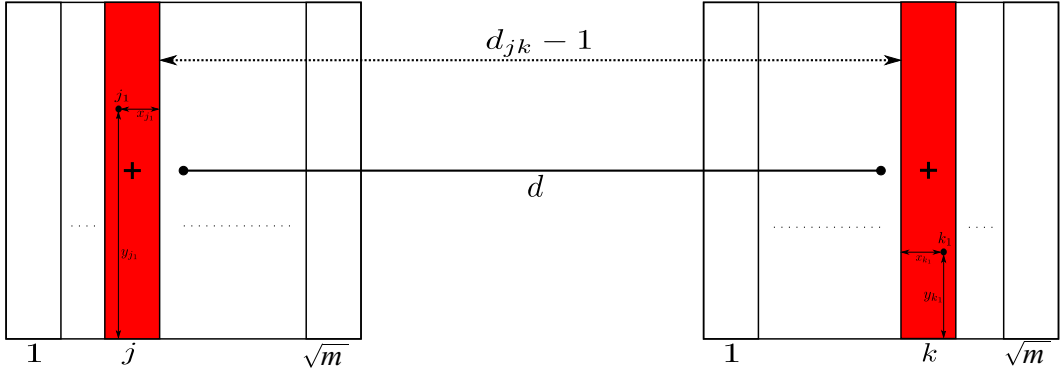}
\caption{Two square clusters with $A=m$ that have a center-to-center distance $d$, with each cluster decomposed into  $\sqrt{m}$ vertical $\sqrt{A}\times 1$ rectangles. $d_{jk}$ is distance between the centers (marked with cross) of the two rectangles $j$ and $k$. Moreover, we have the points $j_1(x_{j_1},y_{j_1})$ and $k_1(x_{k_1},y_{k_1})$ in the rectangles $j$ and $k$, respectively.}
\label{fig:4}
\end{figure}
\begin{proof}[Proof of Lemma \ref{lem:blockNorm}]
Most of the ingredients of the proof come from the proof of the particular case of $A=m$ ($\nu=1$) presented in \cite{HLarxiv1}. In the case of $\nu=1$, the strategy was essentially the following: in order to bound $\Vert \widehat{H} \Vert$, we divide the matrix into smaller blocks, bound the smaller blocks $\Vert {\widehat{H}}_{jk} \Vert$, and apply Lemma \ref{lem:Gersgorin}. We decompose each of the two square clusters into $\sqrt{m}$ vertical $\sqrt{A}\times 1$ rectangles of $\sqrt{m}$ nodes each (See Fig. \ref{fig:4}).

By Lemma \ref{lem:Gersgorin}, we obtain
\begin{equation} \label{eq:ub_H0}
\Vert \widehat{H} \Vert \le \max \left\{ \max_{1 \le j \le \sqrt{m}} \sum_{k=1}^{\sqrt{m}} \Vert \widehat{H}_{jk} \Vert, \max_{1 \le j \le \sqrt{m}} \sum_{k=1}^{\sqrt{m}} \Vert \widehat{H}_{kj} \Vert \right\}
\end{equation}
where the $m \times m$ matrix $\widehat{H}$ is decomposed into blocks $\widehat{H}_{jk}$, $j,k=1,\ldots,\sqrt{m}$, with $\widehat{H}_{jk}$ denoting the $\sqrt{m} \times \sqrt{m}$ channel matrix between $k$-th rectangle of the transmitting cluster and the $j$-th rectangle of the receiving cluster. As shown in Fig. \ref{fig:4}, $d_{jk}$ denotes the corresponding inter-rectangle distance, measured from the centers of the two rectangles. In \cite{HLarxiv1}, it is shown that for $2\sqrt{A}\le d \le A$, where $d$ is the distance between the centers of the two clusters, there exist constants $c,c'>0$ such that
\begin{equation}\label{eq:up_H0jk}
\Vert \widehat{H}_{jk} \Vert^2 \le c' \, \frac{m^{1+\epsilon}}{A\, d_{jk}} \le c \, \frac{m^{1+\epsilon}}{A\, d}
\end{equation}
with high probability as $m \to \infty$. Applying \eqref{eq:ub_H0} and \eqref{eq:up_H0jk}, 
$$
\Vert \widehat{H} \Vert \le \max \left\{ \max_{1 \le j \le \sqrt{m}} \sum_{k=1}^{\sqrt{m}} \Vert \widehat{H}_{jk} \Vert, \max_{1 \le j \le \sqrt{m}} \sum_{k=1}^{\sqrt{m}} \Vert \widehat{H}_{kj} \Vert \right\} \le \left(c\,\frac{m^{2+\epsilon}}{A\, d}\right)^{1/2}.
$$
Therefore, for $\nu=1$ we already have the desired upper bound in \cite{HLarxiv1}. Moreover, to prove the inequality \eqref{eq:up_H0jk}, the authors in \cite{HLarxiv1} use the moments' method, relying on the following inequality:
\begin{align*}
\Vert \widehat{H}_{jk} \Vert^2  & = \lambda_{\max}(\widehat{H}_{jk}\widehat{H}_{jk}^{\dagger}) \le \left( \sum_{k=1}^M (\lambda_k(\widehat{H}_{jk} \widehat{H}_{jk}^{\dagger}))^\ell \right)^{1/\ell}\\
& = \left( \Tr \left((\widehat{H}_{jk} \widehat{H}_{jk}^{\dagger})^\ell \right) \right)^{1/\ell}
\end{align*}
valid for any $\ell \ge 1$. So by Jensen's  inequality, they obtain that $\EE( \Vert \widehat{H}_{jk} \Vert^2) \le \left( \EE( \mathrm{Tr}((\widehat{H}_{jk}\widehat{H}_{jk}^{\dagger})^\ell) ) \right)^{1/\ell}$. Finally, they show that taking $\ell \to \infty$ leads to $\EE( \Vert \widehat{H}_{jk} \Vert^2) \le c \, \frac{\log M}{d_{jk}}$, by precisely showing that 
\begin{equation}\label{eq:final}
\EE( \mathrm{Tr}((\widehat{H}_{jk}\widehat{H}_{jk}^{\dagger})^\ell)= O\left( (\sqrt{m})^{2l} S_l\right) = O\left( \frac{m^{\ell}(c\,\log A)^{\ell-1}}{A^{\ell-1}d_{jk}^{\ell+1}}\right),
\end{equation}
where 
\begin{equation}\label{S_l}
S_{\ell} =|\EE(f_{j_1k_1}f_{j_2k_1}^*\ldots f_{j_{\ell}k_{\ell}}f_{j_1k_{\ell}}^*)|=O\left(\frac{(c\,\log A)^{\ell-1}}{A^{\ell-1}d_{jk}^{\ell+1}}\right),
\end{equation}
with $j_1\neq \ldots \neq j_{\ell}$ and $k_1\neq \ldots \neq k_{\ell}$. Note that $S_l$ does not depend on the particular choice of  $j_1\neq \ldots \neq j_{\ell}$ and $k_1\neq \ldots \neq k_{\ell}$.
This finally implies
\begin{equation*}
\left( \EE( \mathrm{Tr}((\widehat{H}_{jk}\widehat{H}_{jk}^{\dagger})^\ell) ) \right)^{1/\ell} \le \frac{m\,(c\,\log A)^{1-{1/\ell}}}{A^{1-1/{\ell}}d_{jk}^{1+{1/\ell}}} \underset{\ell\rightarrow\infty}{\rightarrow} c\,\frac{m\log A}{A\, d_{jk}}. 
\end{equation*}
The last step includes applying Markov's inequality to get
\begin{align*}
\mathbb{P}\left(\lambda_{\max}(\widehat{H}_{jk}\widehat{H}_{jk}^{\dagger})\geq 
c' \frac{m^{1+\epsilon}}{A\, d_{jk}}\right)&\leq \frac{\EE((\lambda_{\max}(\widehat{H}_{jk}\widehat{H}_{jk}^{\dagger}))^\ell)}
{(c' m^{1+\epsilon}/(A\, d_{jk}))^{\ell}}\\
&\le \frac{A\,(\log A)^{{\ell}-1}}{d_{jk}\, m^{\epsilon\ell}}
\end{align*}
which, for any fixed $\epsilon > 0$, can be made arbitrarily small by taking $\ell$ sufficiently large.

To extend this result for any $\nu>0$, we reuse the upper bound obtained in \cite{HLarxiv1} on $S_{\ell}$. We can show that the upper bound on $S_{\ell}$ is also applied to the case where the $\ell$ points move in a square of area $A$ instead of $\sqrt{A}\times 1$ rectangle. However, we omit this small technical issue to emphasize on the main result. Therefore, from now on $S_{\ell}$ assumes that $\ell$ the points corresponding to $j$'s  and $k$'s are randomly chosen in two squares of area $A$ apart by a distance $d$.

After sketching the proof in \cite{HLarxiv1} for the particular case $\nu=1$, we use the same approach to prove the given Lemma. For $\nu>0$, consider the $m\times m$ channel matrix $\widehat{H}$ between two square clusters of $m$ nodes distributed uniformly at random each of area $A=m^{\nu}$ and separated by distance $2 \sqrt{A} \le d \le A$. 

We have the following inequalities:

For the first moment, we have
$$
\EE\left(\mathrm{Tr}\left(\widehat{H}\widehat{H}^{\dagger}\right)\right)=\sum_{j_1,k_1=1}^{m}\EE(\hat{h}_{j_1k_1}\hat{h}_{j_1k_1}^*) =\sum_{j_1,k_1=1}^{m}\EE(|\hat{h}_{j_1k_1}|^2) =\sum_{j_1,k_1=1}^{m}\frac{1}{r_{j_1k_1}^2}=O\left(\frac{m^2}{d^2}\right).
$$

For the second moment, we have
\begin{align*}
\EE(\mathrm{Tr}((\widehat{H}\widehat{H}^{\dagger})^2))
&=\sum_{\substack{j_1,j_2,k_1,k_2=1}}^{m} \EE(\hat{h}_{j_1k_1}\hat{h}_{j_2k_1}^*\hat{h}_{j_2k_2}\hat{h}_{j_1k_2}^*) \\
&\le\sum_{\substack{j_1=j_2\\k_1,k_2}} \EE(\hat{h}_{j_1k_1}\hat{h}_{j_2k_1}^*\hat{h}_{j_2k_2}\hat{h}_{j_1k_2}^*) +\sum_{\substack{j_1, j_2\\k_1=k_2}} \EE(\hat{h}_{j_1k_1}\hat{h}_{j_2k_1}^*\hat{h}_{j_2k_2}\hat{h}_{j_1k_2}^*) +\sum_{\substack{j_1\neq j_2\\k_1\neq k_2}} \EE(\hat{h}_{j_1k_1}\hat{h}_{j_2k_1}^*\hat{h}_{j_2k_2}\hat{h}_{j_1k_2}^*)\\
&\leq 2\frac{m^3}{d^4} + m^4 S_2 \overset{(a)}{\le} 2\frac{m^3}{d^4} + m^4 \frac{\log A}{A\, d^3}=O\left(\max\left\{ \frac{m^3}{d^4} , \frac{m^4 \, \log A}{A\, d^3}\right\}\right).
\end{align*}

As in \cite{HLarxiv1}, it can be shown that 
\begin{align*}
\EE(\mathrm{Tr}((\widehat{H}\widehat{H}^{\dagger})^{\ell}))&\leq \frac{2\,\ell}{m} \sum_{t=1}^{\lfloor \ell/2\rfloor}
\EE(\mathrm{Tr}((\widehat{H}\widehat{H}^{\dagger})^{t}))\,\EE(\mathrm{Tr}((\widehat{H}\widehat{H}^{\dagger})^{\ell-t})) + m^{2\, \ell} S_{\ell}.
\end{align*}

Using the above inequality, it can be shown that
$$
\EE(\mathrm{Tr}((\widehat{H}\widehat{H}^{\dagger})^{\ell}))=O\left(\max\left\{ \frac{m^{\ell+1}}{d^{2\, \ell}} , \frac{m^{2\, \ell} \, (\log A)^{\ell-1}}{A^{\ell-1}\, d^{\ell+1}}\right\}\right).
$$
Applying the Markov's inequality as above, concludes the proof.
\begin{figure}
\centering
\includegraphics[scale=0.55]{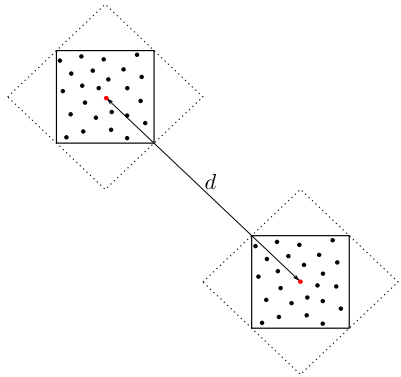}
\caption{Two tilted square clusters that have a center-to-center distance $d$. We can draw larger squares (drawn in dotted line) containing the original clusters with the same centers that are aligned.}
\label{fig:5}
\end{figure}

A last remark is that we proved lemma \ref{lem:blockNorm} for aligned clusters. However, the proof can be easily generalized to tilted clusters, as shown in Fig. \ref{fig:5}. We can always draw a larger cluster containing the original cluster and having the same center. The larger cluster can at most contain twice as many nodes as the original cluster. The large clusters are now aligned. Moreover, the distance $d$ from the centers of the two newly created large clusters still satisfies the required condition ($2\sqrt{A}\le d \le A$). 
\end{proof} 

%\bibliographystyle{IEEEtran}
%\bibliography{BIBfile}

\bibliographystyle{IEEEtran}
\bibliography{BIBfile}

\end{document}